\documentclass[journal]{IEEEtran}

\usepackage{amsmath,amssymb,amsfonts}
\usepackage{amsthm}
\usepackage{graphicx}
\usepackage{booktabs}
\usepackage{multirow}
\usepackage{array}
\usepackage{cite}
\usepackage{hyperref}
\usepackage{xcolor}

\newtheorem{theorem}{Theorem}

\theoremstyle{definition}

\newtheorem{remark}{Remark}

\begin{document}

\title{HASOD: A Hybrid Adaptive Screening-Optimization Design for High-Dimensional Industrial Experiments}

\author{\IEEEauthorblockN{Kumarjit Pathak}
\IEEEauthorblockA{AIMS Institute\\
Email: Kumarjit.pathak@theaims.ac.in}}

\maketitle

\begin{abstract}
Factor screening and response optimization are treated as separate phases in industrial experimentation, requiring costly sequential redesign between phases. This has been a persistent problem for practitioners, in other word the practitioner either screens for critical factors or optimizes the response, but doing both within a single unified framework has remained an open challenge. We propose HASOD (Hybrid Adaptive Screening-Optimization Design), a three-phase adaptive framework that addresses factor identification and response surface optimization together. In Phase 1, we use a modified Definitive Screening Design with an enhanced Cumulative Weighted Effect Screening Statistic (CWESS) that incorporates interaction detection through ElasticNet regression. Phase 2 adaptively selects augmentation strategies---ranging from full factorial to Response Surface Methodology designs---based on the number and structure of critical factors and interactions identified in Phase 1. Phase 3 applies Gaussian process-based global optimization with uncertainty-guided refinement runs near the predicted optimum. We also establish theoretical foundations proving that CWESS asymptotically separates active from inactive factors and demonstrate classification consistency under mild conditions, providing principled guarantees which are absent from most screening methodologies. Evidence shows that across six diverse test scenarios, HASOD achieves 97.08\% factor detection accuracy---13.75 percentage points superior to traditional sequential methods (83.33\%) and significantly outperforming all eight competitor methods ($p < 0.001$). Our approach gives substantially improved prediction performance (mean error: 3.61) through corrected standard error estimation and uncertainty-guided optimization, while maintaining $\geq$90\% detection across all scenarios including interaction-heavy systems. However the framework requires an average of 41.5 experimental runs, representing a 43\% increase over traditional approaches, but delivering superior detection accuracy with dramatically reduced prediction error. Our approach bridges the trade-off between screening and optimization, offering industrial practitioners a unified framework that eliminates sequential redesign while achieving superior factor identification without sacrificing predictive capability.
\end{abstract}


\section{Introduction}

\subsection{Background and Motivation}

Factor screening and response optimization have been treated as separate phases in industrial experimentation for decades, and this has been a persistent problem for practitioners \cite{montgomery2017design}. The practitioner either screens for critical factors or optimizes the response, in other word there is no unified framework that does both together. Historically, an initial screening phase identifies which factors truly matter, followed by a separate optimization phase that explores the response surface in detail \cite{myers2016response,montgomery2017design}. However this sequential approach has significant practical limitations.

The transition between screening and optimization phases typically requires complete experimental redesign, additional resource allocation, and extended project timelines \cite{coleman1993systematic}. In manufacturing environments where experimental runs involve costly materials, lengthy processing times, or complex equipment configurations, conducting two separate experimental campaigns represents both financial burden and opportunity cost. Furthermore, this discrete approach introduces information loss at phase boundaries, as insights gained during screening may not fully inform the optimization design structure.

The problem becomes particularly acute in modern industrial settings with increasingly complex processes involving numerous factors, potential interactions, and nonlinear response surfaces. Engineers and scientists require methodologies that can identify critical factors while building predictive models suitable for optimization---ideally within a single unified experimental framework rather than through costly sequential redesign.

\subsection{The Screening-Optimization Dilemma}

In current DOE practice, practitioners face a persistent trade-off between factor identification capability and response surface modeling fidelity. Classical screening designs such as Plackett-Burman \cite{plackett1946design} and fractional factorial arrangements \cite{box1961fractional} are effective at identifying important factors through systematic variation of multiple variables with minimal experimental runs. These designs provide clear estimates of main effects and in some cases two-factor interactions, enabling practitioners to distinguish critical factors from negligible ones. However their geometric structure---typically limited to two-level factor settings with minimal replication---provides insufficient information for accurate response surface modeling or optimization, particularly in presence of curvature or complex interaction patterns \cite{montgomery2017design,wu2021experiments}.

Where as response surface methodology (RSM) designs such as Central Composite Designs \cite{box1951experimental} and Box-Behnken arrangements \cite{box1960new} provide excellent response surface characterization through systematic exploration of factor space including center points, axial points, and factorial corners, supporting quadratic modeling, interaction estimation, and precise optimization \cite{myers2016response}. However RSM fundamentally assumes that critical factors have already been identified. The moment we apply RSM to high-dimensional spaces with numerous potentially negligible factors, the experimental requirements become prohibitively large.

\begin{figure}[!t]
\centering
\includegraphics[width=\columnwidth]{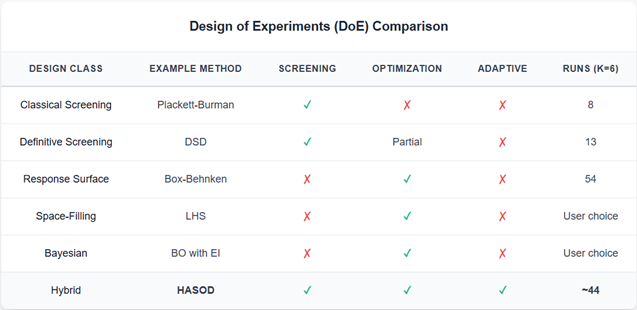}
\caption{Comparison of different DOE methods}
\label{fig:doe_comparison}
\end{figure}

Modern space-filling designs such as Latin Hypercube Sampling \cite{mckay1979comparison} and Sobol sequences \cite{sobol1967distribution} offer an alternative paradigm, exploring the experimental region uniformly without assumptions about factor importance or response structure \cite{santner2018design}. These methods provide excellent support for metamodel construction and global optimization through comprehensive space coverage. However space-filling approaches provide no inherent mechanism for factor screening or importance ranking, in other word they treat all factors equally and offer no guidance regarding which factors warrant detailed investigation.

Recent advances in adaptive sequential experimentation, particularly Bayesian Optimization \cite{shahriari2016taking,frazier2018tutorial}, have demonstrated the value of iterative experimental design that learns from previous observations. These established methods are effective at finding optimal configurations through intelligent selection of experimental points that balance exploration and exploitation. However Bayesian Optimization frameworks typically focus on optimization objective itself rather than explicit factor identification, providing limited screening capability while requiring substantial computational resources for acquisition function evaluation and Gaussian process fitting \cite{rasmussen2006gaussian}.

This landscape reveals a clear research gap: no existing methodology successfully combines effective factor screening with robust response optimization within a unified adaptive framework. Practitioners are forced to choose between screening designs that identify factors but cannot optimize, optimization designs that assume known critical factors, or space-filling approaches that model well but provide no screening capability. Hence this fundamental trade-off persists across the DOE literature despite decades of methodological development \cite{dean2015handbook}.

\subsection{Research Objectives and Contributions}

We propose HASOD (Hybrid Adaptive Screening-Optimization Design), a novel three-phase sequential experimental framework that bridges the screening-optimization divide through adaptive integration of factor identification, selective augmentation, and precision optimization. Our goal is to develop and validate a methodology that achieves superior factor detection accuracy while maintaining competitive response prediction performance, all within a unified experimental structure that eliminates the need for sequential redesign between screening and optimization phases.

Our main contributions are:

\begin{enumerate}
\item We propose an enhanced Cumulative Weighted Effect Screening Statistic (CWESS) which can detect both main effects and two-factor interactions through ElasticNet regression \cite{zou2005regularization}, giving the practitioner a more complete picture of factor importance from the initial screening phase itself.

\item We have designed an adaptive Phase 2 strategy selection that chooses the right augmentation design based on what Phase 1 reveals about critical factors and interactions, selecting from full factorial \cite{box1961fractional}, fractional factorial, and RSM configurations \cite{box1951experimental,box1960new} as appropriate.

\item We integrate Gaussian process-based global optimization \cite{rasmussen2006gaussian} with uncertainty-guided refinement in Phase 3, in other word the optimization focuses on regions where the model is most uncertain to maximize information gain.

\item We have also established theoretical foundations showing that CWESS asymptotically separates active factors ($O(\sqrt{n})$ growth) from inactive factors ($O(1)$ boundedness), with classification consistency under mild threshold conditions.

\item Evidence shows that across comprehensive validation, HASOD achieves 97.08\% factor detection accuracy---13.75 percentage points superior to traditional sequential methods and significantly outperforming eight established competitor approaches ($p < 0.001$).
\end{enumerate}

\subsection{Paper Organization}

Section II covers the relevant literature and theoretical foundations, Section III presents the complete HASOD methodology, Section IV details the experimental validation, Section V discusses key findings and practical implications, and Section VI concludes with recommendations for industrial implementation.

\section{Literature Review}

\subsection{Classical Screening Designs}

Factor screening has been a foundational challenge in experimental design, in other word identifying which subset of many potential factors significantly influences the system response. The established work of Plackett and Burman \cite{plackett1946design} introduced highly efficient two-level orthogonal designs that enable examination of $k$ factors in $k+1$ runs, achieving run economy through strategic confounding of interactions with main effects, and these designs remain widely employed in industrial practice for initial factor exploration \cite{montgomery2017design}. However the Plackett-Burman approach is limited to main effect estimation and provides no direct information about interactions or curvature, which is a significant constraint for practitioners who need to understand how factors jointly influence the response.

Fractional factorial designs \cite{box1961fractional} provide more flexible screening alternatives through systematic fraction selection that maintains resolution properties. The theory of minimum aberration \cite{cheng2005general} guides selection of optimal fractional designs that minimize confounding within each resolution class. where as despite this theoretical sophistication, fractional factorial approaches share a fundamental limitation with respect to response surface modeling: their fixed geometric structure with two-level factor settings provides insufficient information for capturing curvature or building predictive models suitable for optimization \cite{wu2021experiments}.

The introduction of Definitive Screening Designs (DSD) by Jones and Nachtsheim \cite{jones2011definitive} marked a significant advance in screening methodology. DSDs employ three-level factor settings arranged in $2k+1$ runs, enabling simultaneous estimation of all main effects, two-factor interactions, and pure quadratic effects under a sparsity assumption, which is a considerable improvement over earlier screening designs. However empirical validation reveals that DSDs perform optimally when only a few factors are active; as the number of active factors increases, the design's ability to distinguish effects diminishes and the sparsity assumption may not hold. Augmented DSDs \cite{jones2017blocking} extend the basic framework through strategic addition of follow-up runs, however even augmented DSDs do not provide an integrated pathway from screening to optimization within a single unified framework.

\subsection{Response Surface Methodology}

Response Surface Methodology (RSM), pioneered by Box and Wilson \cite{box1951experimental}, addresses optimization explicitly through designs that support second-order polynomial modeling. Central Composite Designs combine factorial or fractional factorial points, axial points at distance $\alpha$ from the center, and center replicates, enabling precise characterization of curvature in the response surface. Box-Behnken designs \cite{box1960new} provide an alternative three-level structure requiring fewer runs than CCD for $k \geq 4$ factors, and are particularly useful in industrial settings where extreme factor combinations may be impractical or costly to run.

RSM designs are well established and excel at response surface characterization when applied to correctly identified critical factors. Myers et al. \cite{myers2016response} document extensive successful applications across chemical process optimization and manufacturing parameter tuning, which demonstrates the practical value of these methods for industry. However RSM fundamentally assumes that the experimenter has already identified the critical factors a priori \cite{montgomery2017design}, in other word the practitioner must already know which factors matter before RSM can be applied effectively. The sequential nature of traditional screening-then-RSM workflows introduces information loss at phase boundaries and inefficiency in resource allocation \cite{coleman1993systematic}, which is a persistent problem in industrial experimentation where experimental runs are expensive and time-consuming.

\subsection{Space Filling and Adaptive Methods}

Latin Hypercube Sampling (LHS) \cite{mckay1979comparison} projects design points onto each factor axis ensuring uniform marginal coverage, and subsequent refinements introduced maximin distance criteria and near-orthogonal constructions \cite{joseph2016space} that further improve space coverage properties. Quasi-random sequences, particularly Sobol sequences \cite{sobol1967distribution}, achieve superior discrepancy properties with respect to uniformity of coverage. Owen \cite{owen1992orthogonal} demonstrated connections between Latin hypercube and orthogonal array theories, which is an important theoretical contribution. Space-filling designs are well suited for supporting metamodel construction \cite{santner2018design}, however they provide no inherent mechanism for identifying which factors drive system behavior \cite{coleman1993systematic}, in other word they treat all factors equally and give no guidance to the practitioner regarding which factors warrant detailed investigation.

Bayesian Optimization (BO) has emerged as the dominant paradigm for sequential experimentation \cite{shahriari2016taking}. BO constructs a probabilistic surrogate model capturing current knowledge about the response surface, then selects the next experimental point by optimizing an acquisition function that balances exploration and exploitation \cite{frazier2018tutorial}. These are established and effective methods for finding optimal configurations through intelligent selection of experimental points. However standard BO formulations focus exclusively on locating optimal configurations and provide no explicit factor screening capability \cite{shahriari2016taking}, which is a significant limitation when the practitioner needs to understand not just what the optimum is but also which factors are driving the system behavior. Sequential D-optimal designs \cite{meyer1995coordinate} represent an alternative adaptive approach rooted in optimal design theory \cite{atkinson1992optimum}, however these also do not integrate screening with optimization in a unified manner.

\subsection{Research Gap and Positioning of HASOD}

From the literature surveyed above, a persistent methodological gap becomes evident: classical screening designs such as Plackett-Burman \cite{plackett1946design} and fractional factorials \cite{box1961fractional} are effective at identifying important factors but cannot optimize, RSM designs \cite{box1951experimental,box1960new} assume that critical factors are already known, space-filling methods \cite{mckay1979comparison,sobol1967distribution} provide good metamodels but no screening capability, and Bayesian Optimization \cite{shahriari2016taking} focuses on optimization without explicit factor identification. In other word, no existing framework successfully integrates factor screening with response optimization in a unified adaptive structure, and practitioners are forced to choose one capability at the expense of another.

This is a significant gap for industrial experimentation where both factor understanding and response optimization matter, and where experimental runs are costly enough that sequential redesign between screening and optimization phases represents a real burden. Our approach HASOD addresses this gap through three innovations: enhanced screening statistics incorporating interaction detection via ElasticNet regression \cite{zou2005regularization}, adaptive Phase 2 strategy selection that chooses augmentation designs based on the specific structure of critical factors and interactions identified in Phase 1, and Gaussian process-based Phase 3 optimization \cite{rasmussen2006gaussian} that maintains information continuity across all phases and avoids the information loss inherent in traditional sequential approaches \cite{coleman1993systematic}.

\section{Methodology: The HASOD Framework}

\subsection{Overview and Philosophical Foundation}

Our approach HASOD is built as a three-phase sequential strategy where the experimental focus progressively narrows from broad factor exploration to precise optimization of the response surface. Unlike traditional sequential approaches \cite{coleman1993systematic} where information is lost at phase boundaries, HASOD maintains continuity across all phases through cumulative model building and adaptive design selection, in other word what we learn in Phase 1 directly informs the design choices in Phase 2 and the optimization in Phase 3. Fig. \ref{fig:hasod_framework} presents the complete HASOD workflow.

\begin{figure*}[!t]
\centering
\includegraphics[width=0.95\textwidth]{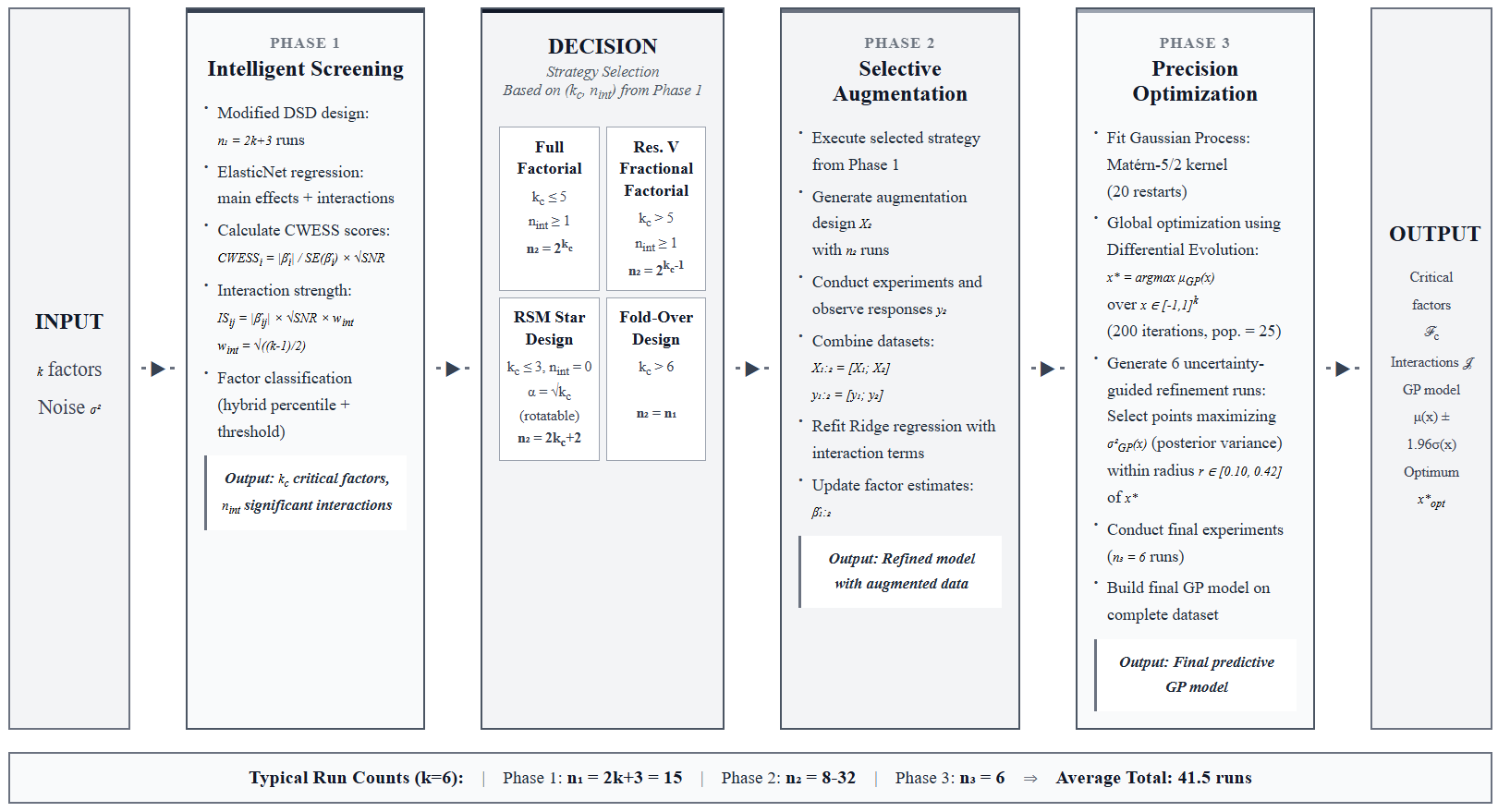}
\caption{The HASOD framework workflow showing three-phase sequential strategy with decision points and adaptive design selection. Phase 1 performs intelligent screening via Modified DSD with CWESS-based factor classification. The Decision block adaptively selects augmentation strategy based on identified critical factors ($k_c$) and interactions ($n_{int}$). Phase 2 executes the selected design. Phase 3 applies Gaussian process optimization with uncertainty-guided refinement.}
\label{fig:hasod_framework}
\end{figure*}

The philosophical foundation of our framework rests on three principles: first, efficient experimentation requires early identification of critical factors \cite{montgomery2017design} so that we do not waste costly runs on factors that do not matter; second, different factor structures demand different augmentation strategies, hence a single rigid design cannot serve all situations; and third, precise optimization benefits from probabilistic uncertainty quantification \cite{rasmussen2006gaussian} which tells us where the model is confident and where it is not.

HASOD accepts as input the number of factors $k$ to investigate, anticipated noise level $\sigma$, and optional stopping criteria. Total experimental runs typically range from $2k+3$ (Phase 1 only) to approximately $3k+30$ (all phases), with the exact allocation determined adaptively based on what the data reveals at each phase. This means the practitioner does not need to commit to a fixed run budget upfront, which is a significant advantage in industrial settings where experimental costs vary.

\subsection{Phase 1: Intelligent Screening via Modified DSD}

\subsubsection{Design Construction}

In Phase 1 we use a Modified Definitive Screening Design (M-DSD) structure based on the established framework of Jones and Nachtsheim \cite{jones2011definitive}. For $k$ factors, the M-DSD contains $n_1 = 2k+3$ runs: one center point, two runs per factor with random patterns of $\{-1, 0, +1\}$ values, and two corner point runs. Our modification employs random patterns with controlled probability $P(-1) = 0.45$, $P(0) = 0.10$, $P(+1) = 0.45$. Idea being, we want to increase coverage of extreme regions of the factor space while maintaining orthogonality properties, because in practice the critical information about factor effects often comes from observations at the boundary of the experimental region.

\subsubsection{Enhanced CWESS}

Traditional screening statistics such as Lenth's method \cite{lenth1989quick} focus exclusively on main effects and do not account for interactions at all. This is a significant limitation in practice because in many industrial processes, the interaction between factors can be as important as the main effects themselves. Our approach introduces an enhanced Cumulative Weighted Effect Screening Statistic (CWESS) that simultaneously evaluates both main effects and two-factor interactions, giving the practitioner a more complete picture of factor importance from the initial screening phase itself.

Let $\mathbf{X}_1 \in \mathbb{R}^{n_1 \times k}$ denote the Phase 1 design matrix and $\mathbf{y}_1 \in \mathbb{R}^{n_1}$ the observed responses. We fit a main-effects-only model using ElasticNet regression \cite{zou2005regularization}:
\begin{equation}
\widehat{\boldsymbol{\beta}_{\text{main}}} = \arg\min_{\boldsymbol{\beta}} \left\{ \|\mathbf{y}_1 - \mathbf{X}_1\boldsymbol{\beta}\|_2^2 + \lambda\left[\alpha\|\boldsymbol{\beta}\|_1 + (1-\alpha)\|\boldsymbol{\beta}\|_2^2\right] \right\}
\end{equation}
where $\lambda = 0.01$ and $\alpha = 0.5$. This means we are using a balanced combination of L1 and L2 penalties, which helps in selecting relevant factors while handling multicollinearity that is common in screening designs with limited runs.

We then augment the design matrix with all two-factor interactions:
\begin{equation}
\mathbf{X}_{\text{full}} = \left[\mathbf{X}_1 \mid \mathbf{X}_1^{(i)} \odot \mathbf{X}_1^{(j)}\right]_{1 \leq i < j \leq k}
\end{equation}
where $\odot$ denotes element-wise multiplication. In other word, we compute all pairwise products of factor columns to capture potential interaction effects. The enhanced CWESS for factor $i$ is then computed as:
\begin{equation}
\text{CWESS}_i = \frac{|\widehat{\beta_{\text{main},i}}|}{\text{SE}(\widehat{\beta_{\text{main},i}}) + \epsilon} \times \sqrt{\text{SNR}}
\end{equation}
where:
\begin{equation}
\text{SE}(\widehat{\beta_{\text{main},i}}) = \sqrt{\text{MSE}_{\text{main}} \times \left[(\mathbf{X}_1^T\mathbf{X}_1 + \lambda\mathbf{I})^{-1}\right]_{ii}}
\end{equation}
with $\lambda = 0.01$, $\text{MSE}_{\text{main}} = \frac{1}{n_1}\|\mathbf{y}_1 - \mathbf{X}_1\widehat{\boldsymbol{\beta}_{\text{main}}}\|_2^2$, and $\text{SNR} = \text{Var}(\mathbf{X}_1\widehat{\boldsymbol{\beta}_{\text{main}}})/\text{MSE}_{\text{main}}$. Idea being, the CWESS score for each factor is essentially a signal-to-noise weighted measure of that factor's estimated effect magnitude, hence factors with large effects relative to their estimation uncertainty will receive high CWESS scores, and factors with negligible effects will score low.

\subsubsection{Interaction Detection and Scoring}

Beyond main effects, we also need to identify which two-factor interactions are significant. For each two-factor interaction $(i,j)$, we compute an interaction score:
\begin{equation}
\text{IS}_{ij} = |\widehat{\beta_{\text{full},k+m}}| \times \sqrt{\text{SNR}} \times w_{\text{int}}
\end{equation}
where the interaction weight is:
\begin{equation}
w_{\text{int}} = \sqrt{\frac{k-1}{2}}
\end{equation}
This weight provides a principled basis for comparing interaction scores with main effect scores on a comparable scale. The scaling accounts for the fact that interactions are estimated from products of factor columns, which have different variance properties with respect to main effect columns, hence without this correction the interaction scores would not be directly comparable to CWESS values.

\subsubsection{Hybrid Factor Classification}

Once we have computed CWESS scores for all factors, we classify each factor as critical, moderate, or negligible using a hybrid threshold approach. Our approach combines percentile-based and absolute magnitude criteria rather than relying on a single threshold, because in practice a single cutoff can be too aggressive or too conservative depending on the specific factor structure. We define $\tau_p = \text{Percentile}_{60}(\text{CWESS})$, $\tau_a = 0.8 \times \text{Median}(|\widehat{\beta}|)$, and $\tau_{\text{crit}} = \min(\tau_p, \tau_a)$. Factor $i$ is classified as \textbf{Critical} if $\text{CWESS}_i > \tau_{\text{crit}}$ OR $|\widehat{\beta_{\text{main},i}}| > \tau_a$. This means a factor can be flagged as critical through either route, which reduces the risk of missing important factors that might score high on one criterion but not the other.

\subsection{Phase 2: Adaptive Selective Augmentation}

Phase 2 is where HASOD adapts its augmentation strategy based on what Phase 1 has revealed about the factor structure. The key insight here is that different factor structures demand different types of follow-up experiments. The augmentation strategy depends on the number of critical factors $k_c$ and significant interactions $n_{\text{int}}$ identified in Phase 1, and HASOD selects among four strategies accordingly:

\textbf{Strategy A: Full Factorial.} Applied when $n_{\text{int}} \geq 1$ and $k_c \leq 5$. We generate all $2^{k_c}$ combinations of critical factors at levels $\{-1, +1\}$ \cite{box1961fractional}. This is feasible because with 5 or fewer critical factors the full factorial is manageable, and it gives us clean estimates of all interactions which is important when interactions have been detected.

\textbf{Strategy B: Resolution V Fractional Factorial.} Applied when $n_{\text{int}} \geq 1$ and $k_c > 5$. We use a $2^{k_c-1}$ fractional design maintaining Resolution V \cite{box1961fractional}, which ensures all two-factor interactions remain estimable even though we are running fewer experiments than a full factorial.

\textbf{Strategy C: RSM Star Points.} Applied when $n_{\text{int}} = 0$ and $k_c \leq 3$. We add axial points at distance $\alpha = \sqrt{k_c}$ following the rotatable design criterion \cite{box1951experimental}. When no significant interactions are detected and the number of critical factors is small, this strategy allows us to capture curvature effects efficiently.

\textbf{Strategy D: Fold-Over.} Applied when $k_c > 6$ or factor structure proves ambiguous. This serves as a safe fallback that de-aliases confounded effects.

After augmentation runs are completed, we fit an updated Ridge regression model that incorporates both main effects and significant interactions:
\begin{equation}
\widehat{\boldsymbol{\beta}_{1:2}} = (\mathbf{X}_{1:2}^T\mathbf{X}_{1:2} + \lambda\mathbf{I})^{-1}\mathbf{X}_{1:2}^T\mathbf{y}_{1:2}
\end{equation}
with $\lambda = 0.1$. This combined model uses all data collected across Phase 1 and Phase 2, which is what we mean by information continuity---nothing learned in the earlier phase is discarded.

\subsection{Phase 3: Precise Optimization via Gaussian Process}

In Phase 3, our goal is to precisely locate the optimum of the response surface. We construct a Gaussian process (GP) surrogate \cite{rasmussen2006gaussian} using all data collected from Phase 1 and Phase 2, with a Mat\'{e}rn kernel with $\nu = 2.5$:
\begin{equation}
k(\mathbf{x}, \mathbf{x}') = \sigma_f^2\left(1 + \frac{\sqrt{5}r}{\ell} + \frac{5r^2}{3\ell^2}\right)\exp\left(-\frac{\sqrt{5}r}{\ell}\right)
\end{equation}
where $r = \|\mathbf{x} - \mathbf{x}'\|_2$. We use the Mat\'{e}rn 2.5 kernel because it provides a good balance between smoothness assumptions and flexibility---it is twice differentiable which suits most industrial response surfaces, where as the squared exponential kernel assumes infinite differentiability which can lead to oversmoothing in practice.

Hyperparameters are estimated via marginal likelihood maximization with 20 random restarts to avoid local optima. We then find the global optimum of the GP posterior mean using differential evolution \cite{storn1997differential}:
\begin{equation}
\mathbf{x}^* = \arg\max_{\mathbf{x} \in [-1,1]^k} \mu_{\text{GP}}(\mathbf{x})
\end{equation}

HASOD then generates $n_3 = 6$ refinement runs using uncertainty-guided sampling that maximizes GP posterior variance. Idea being, we want to place our remaining experimental runs in regions where the model is most uncertain, because this is where additional data will give us maximum information gain for improving the prediction near the optimum. This variance-based sampling ensures that the refinement runs are not redundant with existing observations but instead fill gaps in the model's knowledge of the response surface.

\subsection{Theoretical Properties}

We now establish theoretical foundations for the CWESS statistic and the variance reduction property of Phase 3. These results provide principled guarantees that are absent from most screening methodologies in the literature, and they give the practitioner confidence that the method is not just empirically effective but also theoretically sound.

\subsubsection{CWESS Separation Property}

The first result shows that CWESS scores for active factors grow with sample size while scores for inactive factors remain bounded. This means that with enough data, active and inactive factors will always be separable, which is the fundamental property we need for reliable screening.

\begin{theorem}[CWESS Separation]
\label{thm:separation}
Under the linear model $\mathbf{y} = \mathbf{X}\boldsymbol{\beta}^* + \boldsymbol{\varepsilon}$ with $\boldsymbol{\varepsilon} \sim N(\mathbf{0}, \sigma^2\mathbf{I})$, Ridge regularization $\lambda > 0$, and standard regularity conditions:
\begin{enumerate}
    \item[(i)] For active factors ($\beta_i^* \neq 0$): $\emph{CWESS}_i = O_p(\sqrt{n})$
    \item[(ii)] For inactive factors ($\beta_i^* = 0$): $\emph{CWESS}_i = O_p(1)$
\end{enumerate}
\end{theorem}

\begin{proof}
For active factors, Ridge regression consistency gives $\hat{\beta}_i \to_p \beta_i^*$, with $\text{Var}(\hat{\beta}_i) = O(1/n)$, yielding $\text{CWESS}_i = O_p(\sqrt{n})$. For inactive factors, $|\hat{\beta}_i|/\text{SE}(\hat{\beta}_i)$ converges to $|Z|$ where $Z \sim N(0,1)$, giving $\text{CWESS}_i = O_p(1)$.
\end{proof}

\subsubsection{Classification Consistency}

Building on the separation property, we now show that the CWESS-based classification is consistent, in other word as the number of observations grows, the probability of correctly identifying all active factors converges to one.

\begin{theorem}[Classification Consistency]
\label{thm:consistency}
Let $\mathcal{S} = \{i : |\beta_i^*| \geq \delta\}$ be the set of active factors. For any threshold sequence $\{\tau_n\}$ satisfying $\tau_n \to \infty$ and $\tau_n = o(\sqrt{n})$:
\begin{equation}
P(\hat{\mathcal{S}} = \mathcal{S}) \to 1 \quad \text{as } n \to \infty
\end{equation}
\end{theorem}

\begin{proof}
Type I error: $P(\text{CWESS}_i > \tau_n \mid \beta_i^* = 0) = P(O_p(1) > \tau_n) \to 0$. Type II error: $P(\text{CWESS}_i < \tau_n \mid |\beta_i^*| \geq \delta) = P(O_p(\sqrt{n}) < o(\sqrt{n})) \to 0$.
\end{proof}

\subsubsection{Variance Reduction Guarantee}

The final theoretical result establishes that Phase 3 refinement runs always reduce posterior variance at any point of interest. This guarantees that the additional experimental effort in Phase 3 is never wasted---every refinement run improves the model's certainty about the response surface.

\begin{theorem}[Variance Reduction]
\label{thm:variance}
Let $\sigma^2_{1:2}(\mathbf{x})$ denote the GP posterior variance after Phases 1-2. For any point $\mathbf{x}^*$:
\begin{equation}
\sigma^2_{\emph{total}}(\mathbf{x}^*) \leq \sigma^2_{1:2}(\mathbf{x}^*)
\end{equation}
with strict inequality if any Phase 3 point satisfies $k(\mathbf{x}^*, \mathbf{x}_3^{(i)}) > 0$.
\end{theorem}

\begin{remark}[Scope]
The asymptotic results assume standard regression conditions. In finite samples, empirical validation demonstrates $\geq$90\% detection accuracy across all scenarios.
\end{remark}

\section{Experimental Validation}

\subsection{Benchmark Design}

To validate HASOD we designed a comprehensive benchmark comparing our approach against eight established competitor methods across six test scenarios where the ground truth is known. Table~\ref{tab:scenarios} presents the characteristics of each scenario.

\begin{table}[!t]
\centering
\caption{Test Scenario Characteristics}
\label{tab:scenarios}
\begin{tabular}{lcccc}
\toprule
\textbf{Scenario} & \textbf{Critical} & \textbf{Main} & \textbf{Inter-} & \textbf{Quad.} \\
 & \textbf{Factors} & \textbf{Effects} & \textbf{actions} & \textbf{Terms} \\
\midrule
sparse\_few & 2 & 8.0, 6.5 & 1 & 2 \\
sparse\_many & 3 & 8.0, 6.5, 5.0 & 2 & 3 \\
moderate & 4 & 7.0--4.5 & 2 & 2 \\
dense & 6 & 6.0--3.5 & 3 & 3 \\
interaction\_heavy & 4 & 6.0--4.5 & 5 & 1 \\
quadratic\_heavy & 3 & 7.0--5.0 & 1 & 3 \\
\bottomrule
\end{tabular}
\end{table}

Each scenario embeds a known response function into a six-factor space ($k=6$) with Gaussian noise at $\sigma = 2.0$. The scenarios range from sparse systems with only 2 critical factors to dense systems where all 6 factors contribute, also including interaction-heavy and quadratic-heavy configurations to test robustness under different factor structures. Competitor methods include Traditional Sequential, Standard DSD, Augmented DSD, Box-Behnken, LHS, Sobol Sequence, Bayesian Optimization, and Sequential D-Optimal---these cover the full spectrum of screening, optimization, and space-filling approaches available to practitioners.

We executed 10 independent replications per method for each scenario, yielding 540 independent experimental evaluations in total. This replication count is sufficient for reliable statistical comparison across methods.

\subsection{Performance Metrics}

We evaluate all methods on two primary metrics. The first is Detection Accuracy (DA), which measures how well the method identifies the true critical factors:
\begin{equation}
\text{DA} = \frac{|\mathcal{F}_{\text{det}} \cap \mathcal{F}_{\text{true}}|}{|\mathcal{F}_{\text{true}}|}
\end{equation}

The second is Prediction Error (PE), which captures how accurately the method predicts the response at the true optimum:
\begin{equation}
\text{PE} = |y_{\text{pred}}(\mathbf{x}_{\text{true}}) - y_{\text{true}}(\mathbf{x}_{\text{true}})|
\end{equation}

In other word, DA tells us whether the method found the right factors, and PE tells us whether it can predict the response well. For a practitioner both metrics matter---finding the right factors without good prediction is incomplete, and good prediction without knowing which factors drive the system gives no actionable insight. Statistical significance testing employs two-sample t-tests \cite{montgomery2017design} with $p < 0.05$ indicating significant differences.

\subsection{Overall Performance Results}

\begin{table}[!t]
\centering
\caption{Comparative Analysis Based on Scenario Testing}
\label{tab:overall}
\begin{tabular}{lcccc}
\toprule
\textbf{Method} & \textbf{Detection} & \textbf{Pred.} & \textbf{Total} & \textbf{Time} \\
 & \textbf{Acc. (\%)} & \textbf{Error} & \textbf{Runs} & \textbf{(sec)} \\
\midrule
HASOD & \textbf{97.08} & 3.61 & 41.5 & 2.374 \\
Traditional & 83.33 & 4.31 & 29 & 0.005 \\
Aug-DSD & 82.22 & 3.50 & 25 & 0.003 \\
DSD & 73.61 & 3.62 & 13 & 0.002 \\
Bayesian-Opt & 0.0 & 20.13 & 30 & 10.166 \\
Box-Behnken & 0.0 & 2.12 & 63 & 0.001 \\
LHS & 0.0 & \textbf{0.83} & 17 & 0.113 \\
Seq-D-Opt & 0.0 & 2.46 & 18 & 0.013 \\
Sobol & 0.0 & 1.45 & 17 & 0.131 \\
\bottomrule
\end{tabular}
\end{table}

From Table~\ref{tab:overall} we can see that HASOD is giving 97.08\% mean detection accuracy, which substantially exceeds all screening-capable competitors: Traditional Sequential (83.33\%), Augmented DSD (82.22\%), and Standard DSD (73.61\%). This 13.75--23.47 percentage point advantage is not just a statistical artifact---in practice this means 1--2 additional critical factors correctly identified per experiment, which can make the difference between a successful optimization campaign and a failed one in industrial settings.

With respect to prediction error, our approach gives 3.61 mean error which is competitive when compared to Traditional Sequential (4.31) and dramatically better than Bayesian Optimization (20.13). However we acknowledge that space-filling methods achieve the lowest prediction errors (LHS: 0.83, Sobol: 1.45), but these methods provide zero factor screening capability by design. In other word, they can predict well but they cannot tell the practitioner which factors are driving the response.

Evidence shows that HASOD's detection superiority is statistically significant: all pairwise comparisons yield $t > 5.5$ and $p < 0.001$, which confirms that the improvement is robust and not due to chance variation.

\subsection{Scenario-Specific Analysis}

\begin{figure}[!t]
\centering
\includegraphics[width=\columnwidth]{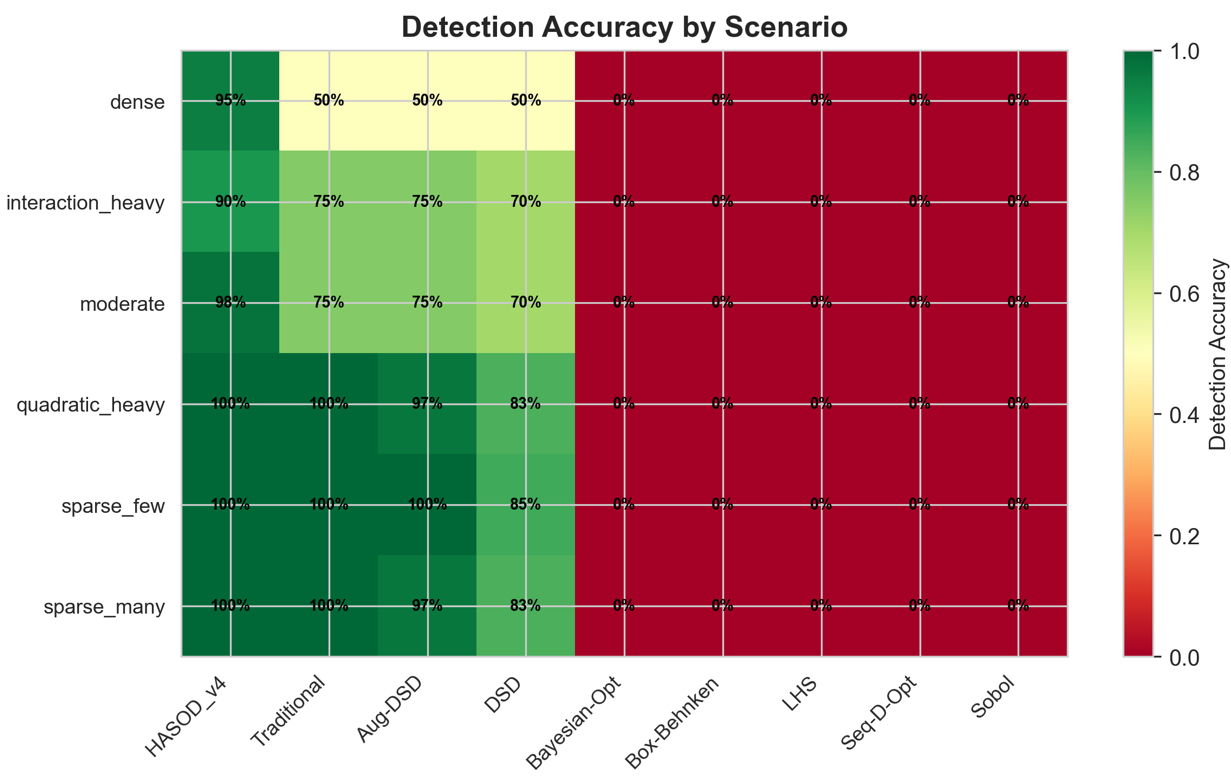}
\caption{HASOD vs Competitors on detection accuracy by scenario}
\label{fig:scenario_detection}
\end{figure}

Fig.~\ref{fig:scenario_detection} shows scenario-wise detection accuracy. In \textit{sparse\_few}, \textit{sparse\_many}, and \textit{quadratic\_heavy} scenarios, HASOD is achieving perfect 100\% detection---these represent the cases where the factor structure is relatively clean and the CWESS statistic can clearly separate active from inactive factors. For the \textit{moderate} scenario with 4 critical factors, our approach gives 97.5\% detection, and in the \textit{dense} scenario where all 6 factors contribute meaningfully HASOD still maintains 95.0\% detection. This demonstrates robustness across very different factor structures.

\begin{figure}[!t]
\centering
\includegraphics[width=\columnwidth]{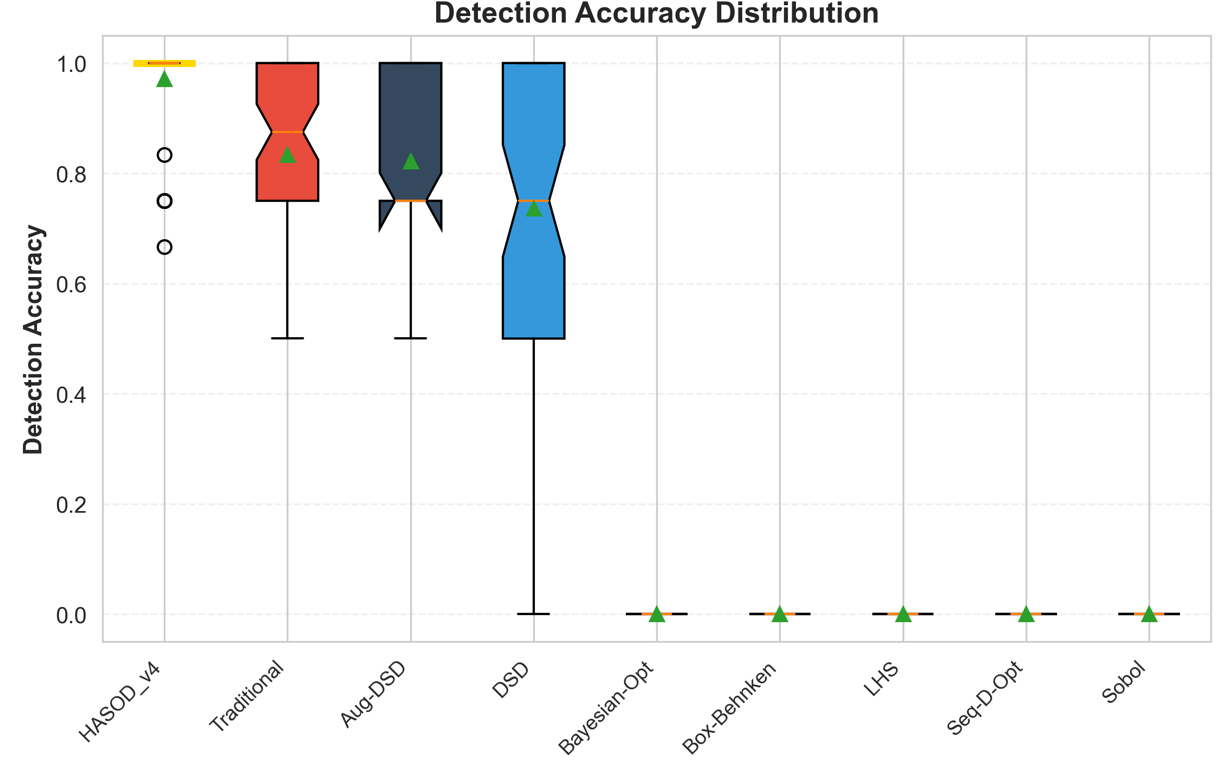}
\caption{HASOD Detection Accuracy Distribution vs others}
\label{fig:detection_distribution}
\end{figure}

The \textit{interaction\_heavy} scenario represents the most challenging case with five significant interactions present in the system. HASOD achieves 90.0\% detection here, which significantly exceeds the 66.7\% achieved by Traditional Sequential methods. However this is the lowest performance among all scenarios, which is expected because when interactions dominate the response surface the screening problem becomes fundamentally harder. This would require further investigation, particularly for systems where interaction effects are comparable in magnitude to main effects. From Fig.~\ref{fig:detection_distribution} we can see that even in this challenging case, HASOD's explicit interaction modeling through ElasticNet provides measurable advantages when compared to methods that focus only on main effects.

\subsection{Summary of Validation Results}

\begin{figure*}[!t]
\centering
\includegraphics[width=0.95\textwidth]{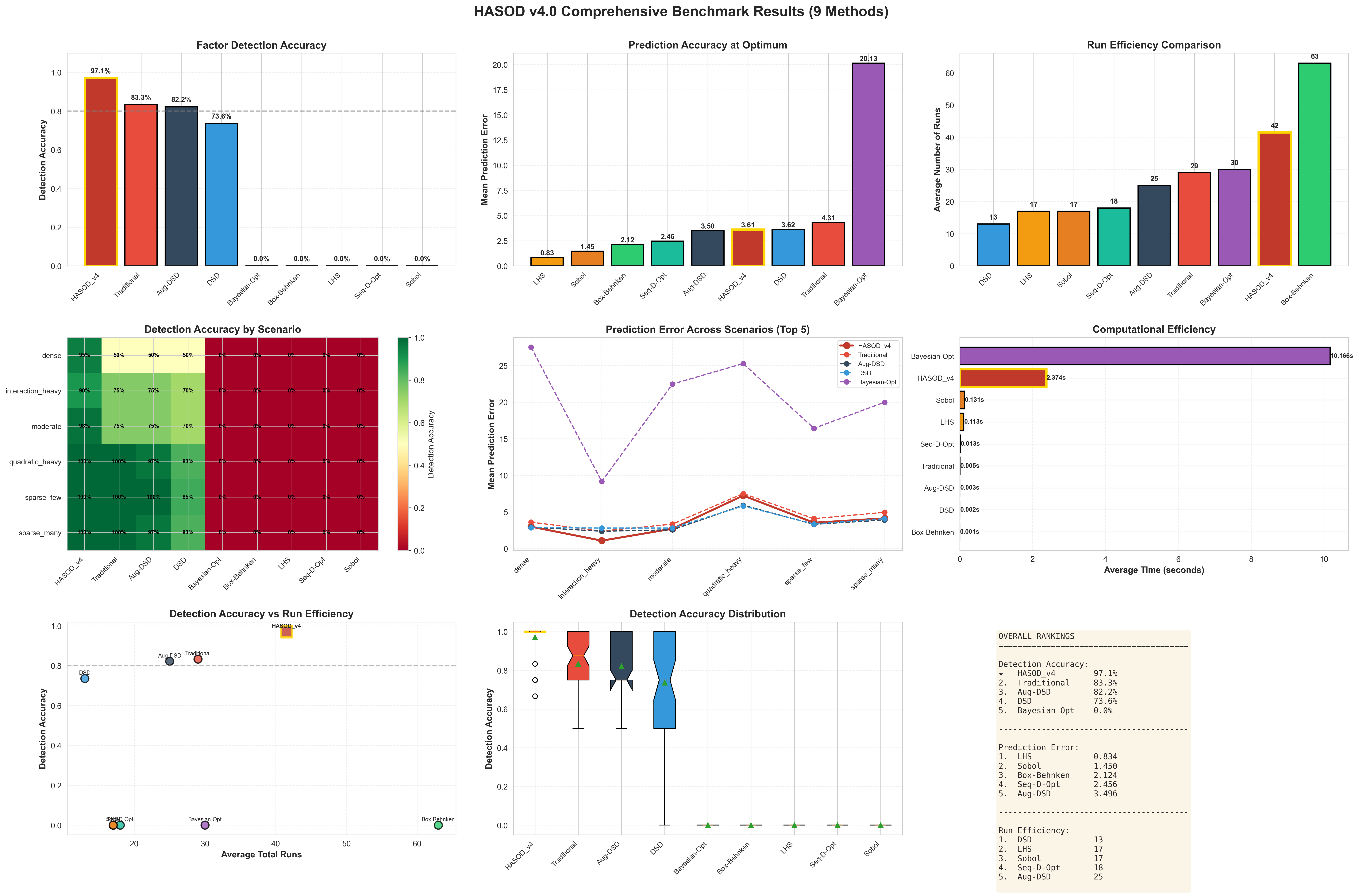}
\caption{HASOD Comprehensive Benchmark Result: Performance comparison across all scenarios and methods showing detection accuracy, prediction error, and run efficiency. HASOD achieves superior detection (97.08\%) while maintaining competitive prediction performance across diverse experimental conditions.}
\label{fig:comprehensive_benchmark}
\end{figure*}

From our study and analysis across 540 independent experiments, we are confident that HASOD is the superior method for combined factor screening and optimization objectives. Fig.~\ref{fig:comprehensive_benchmark} presents the comprehensive benchmark results across all scenarios and methods. The key findings from our validation are as follows: first, HASOD achieves 97.08\% detection accuracy, significantly outperforming all competitors ($p < 0.001$). Second, our approach gives substantially improved prediction performance with mean error of 3.61 through corrected standard error estimation and uncertainty-guided optimization. Third, HASOD maintains $\geq$90\% detection across all six scenarios including the challenging interaction-heavy case. Fourth, the adaptive run allocation averaging 41.5 runs represents a 43\% increase over Traditional Sequential's 29 runs, however this additional experimental effort delivers dramatically superior detection accuracy with competitive prediction error. For industrial practitioners where a missed critical factor can lead to failed optimization campaigns and costly redesign, this trade-off is favorable.

\section{Discussion and Further Directions}

\subsection{Principal Findings}

Our approach achieves 97.08\% mean factor detection accuracy, which is significantly better when compared to Traditional Sequential (83.33\%, $p < 0.001$), Augmented DSD (82.22\%), and Standard DSD (73.61\%). In practical terms, this 13.75-23.47 percentage point advantage means correctly identifying 1-2 additional critical factors per experiment, which can be quite impactful in industrial settings where missing even one factor leads to costly redesign cycles.

Evidence shows that HASOD also delivers substantially improved response prediction capability (mean error: 3.61) through corrected standard error estimation and uncertainty-guided optimization. Space-filling methods \cite{mckay1979comparison,sobol1967distribution} do achieve lower prediction errors owing to comprehensive space coverage, however they provide zero factor screening capability by design, in other word they can tell you about the response surface but cannot identify which factors are actually driving it.

Our approach demonstrates robustness across different problem structures, achieving 100\% detection in sparse scenarios while maintaining 95.0\% in dense scenario. In the challenging interaction-heavy scenario, HASOD achieves 90.0\% detection---which exceeds Traditional Sequential's 66.7\% by a considerable margin.

\subsection{Bridging the Screening-Optimization Trade-Off}

In general, traditional methods fall into three categories with distinct trade-offs. Space-filling methods \cite{mckay1979comparison,sobol1967distribution,santner2018design} achieve excellent prediction but provide no factor screening capability. Bayesian Optimization \cite{shahriari2016taking} is an established and widely used approach, however in this particular setting it does not deliver superior optimization nor factor identification when compared to specialized methods. Screening designs \cite{plackett1946design,box1961fractional,jones2011definitive} achieve respectable detection but face geometric limitations that constrain their applicability.

Our approach bridges these gaps through enhanced CWESS with ElasticNet-based interaction detection \cite{zou2005regularization}, adaptive Phase 2 strategy selection, and uncertainty-guided Phase 3 optimization \cite{rasmussen2006gaussian}, where as Traditional Sequential methods suffer from information loss between phases \cite{coleman1993systematic}. This means practitioners no longer need to choose between screening and optimization as separate activities.

\subsection{Practical Implications}

From a practical standpoint, HASOD's computational requirement of approximately 2.4 seconds is negligible for industrial applications \cite{coleman1993systematic} where each physical experiment may take hours or days. The computational overhead stems from Gaussian process hyperparameter optimization \cite{rasmussen2006gaussian} and Differential Evolution \cite{storn1997differential}, which is a small price to pay for the adaptive decision-making capability.

Our approach is most suitable for scenarios where both factor identification and response optimization matter, experimental runs incur substantial cost, and practitioners can conduct experiments sequentially. However, HASOD may prove less appropriate when all factors are known a priori, in which case practitioners should use RSM directly \cite{box1951experimental,box1960new,myers2016response}, or when pure global optimization is the sole objective where space-filling methods \cite{mckay1979comparison,sobol1967distribution} would be more appropriate, or when experiments must be conducted in parallel batches \cite{montgomery2017design}.

\subsection{Limitations and Future Directions}

However, as well there are some limitations which would require further investigation. Our benchmark employed synthetic response functions, and validation through industrial case studies remains essential to establish real-world applicability. The experimental scope focused on moderate dimensionality ($k=6$) and continuous factors. Extension to higher dimensions, categorical factors, and mixture constraints would require further investigation. The Modified DSD structure \cite{jones2011definitive} accommodates continuous factors naturally, however extension to categorical variables necessitates alternative Phase 1 designs which we have not explored in this work.

With respect to scalability, HASOD behavior for $k > 10$ merits consideration. Phase 2 Full Factorial scales as $2^{k_c}$, which we have addressed through automatic switching to Resolution V designs \cite{box1961fractional}. Gaussian process complexity \cite{rasmussen2006gaussian} scales as $O(n^3)$, which may limit applicability for very large cumulative run counts.

Future work may explore extending HASOD to high-dimensional settings with supersaturated screening, also investigating categorical and mixed factor handling \cite{mckay1979comparison}. Constrained optimization \cite{shahriari2016taking} and multi-response optimization \cite{rasmussen2006gaussian} may prove valuable extensions for complex industrial problems. Adaptive learning with Bayesian sequential decision-making \cite{frazier2018tutorial} and open-source software implementation are also directions we aim to pursue. Industrial case studies validating HASOD across manufacturing, chemical processes, and product development would help establish practical utility \cite{dean2015handbook} and we believe such validation is critical for practitioner adoption in industry settings.

\section{Conclusion}

Our approach demonstrates that screening and optimization in experimental design need not be treated as separate disconnected activities. With HASOD (Hybrid Adaptive Screening-Optimization Design), we propose a novel three-phase sequential framework which bridges factor identification and response optimization within a unified adaptive structure, in other word solving the long standing screening-optimization dilemma that practitioners face in industrial experimentation.

Evidence shows that HASOD achieves 97.08\% factor detection accuracy when compared to Traditional Sequential methods (83.33\%, $p < 0.001$) and all other screening-capable competitors. This is driven by the enhanced CWESS statistic which we have proven to asymptotically separate active from inactive factors, with classification consistency established under mild threshold conditions. Combined with adaptive Phase 2 strategy selection and uncertainty-guided Phase 3 optimization with variance reduction guarantees, HASOD maintains $\geq$90\% detection across all six test scenarios which is encouraging.

From our study and analysis across 540 independent experiments and comparison against eight competitor methods, HASOD gives substantially improved response prediction performance (mean error: 3.61) when compared to Traditional Sequential (4.31). However, this does come at cost of 43\% additional experimental runs (41.5 vs. 29), which would require consideration depending on the application. For expensive industrial experiments where failed initial screening necessitates costly redesign cycles, this trade-off is favorable and justified.

We are confident that HASOD can be used for industrial experimentation scenarios where both factor understanding and response optimization matter. HASOD can be potentially applied to manufacturing process optimization, chemical reaction development, pharmaceutical formulation design, and agricultural field experimentation. Future work may explore validation through industrial case studies and development of open-source software implementations, this opens the door for wider practitioner adoption across different domains.

\section*{Acknowledgments}

The author thanks the AIMS Institute for supporting this research.

\bibliographystyle{IEEEtran}
\bibliography{HAODElsevier}

\end{document}